\documentclass[conference]{IEEEtran}
\usepackage{booktabs}
\IEEEoverridecommandlockouts

\usepackage{pgfplots}
\usepackage{pgfplotstable}
\pgfplotsset{compat=1.18}
\usepgfplotslibrary{groupplots}
\usepackage{cite}
\usepackage{amsmath,amssymb,amsfonts}
\usepackage{graphicx}
\usepackage{textcomp}
\usepackage{subcaption}
\usepackage{xcolor}
\usepackage{amsthm}
\usepackage{algorithmicx}
\usepackage{algorithm}
\usepackage[noend]{algpseudocode}
\usepackage{setspace}
\usepackage[export]{adjustbox}
\usepackage[english]{babel}
\usepackage{enumitem}
\usepackage{multirow}
\usepackage{array} 

\newtheorem{definition}{Definition} 
\theoremstyle{plain}
\usepackage{booktabs}
\usepackage{makecell}
\usepackage{pifont}
\newcommand{\cmark}{\ding{51}}  
\newcommand{\xmark}{\ding{55}}  

\newtheorem{lemma}{Lemma}  

\definecolor{greencustom}{HTML}{8DB4AD}
\def\BibTeX{{\rm B\kern-.05em{\sc i\kern-.025em b}\kern-.08em
    T\kern-.1667em\lower.7ex\hbox{E}\kern-.125emX}}
\begin{document}

\title{Length-Matching Routing for Programmable Photonic Circuits Using Best-First Strategy

}
\author{
    \IEEEauthorblockN{Xiaoke Wang, Dirk Stroobandt}
    \IEEEauthorblockA{\textit{Ghent University} \\ 
    Belgium\\
    \{xiaoke.wang, dirk.stroobandt\}@ugent.be}
}

\maketitle
\begin{abstract}
In the realm of programmable photonic integrated circuits (PICs), precise wire length control is crucial for the performance of on-chip programmable components such as optical ring resonators, Mach-Zehnder interferometers, and optical true time-delay lines. Unlike conventional routing algorithms that prioritize shortest-path solutions, these photonic components require exact-length routing to maintain the desired optical properties.

To address these challenges, this paper presents different length-matching routing strategies to find exact-length paths while balancing search space and runtime efficiently. We propose a novel admissible heuristic estimator and a pruning method, designed to enhance the accuracy and efficiency of the search process. The algorithms are derived from the Best-First search with modified evaluation functions. For two-pin length-matching routing, we formally prove that the proposed algorithms are complete under monotonic heuristics. For multi-pin length-matching challenges, we introduce a pin-ordering mechanism based on detour margins to reduce the likelihood of prematurely blocking feasible routes. Through evaluations on various length-matching benchmarks, we analyze runtime and heuristic performance, demonstrating the effectiveness of the proposed approaches across different layout scenarios.
\end{abstract}




\section{Introduction}
\label{sec:intro}

\begin{figure*}
    \centering
    \includegraphics[width=.9\linewidth]{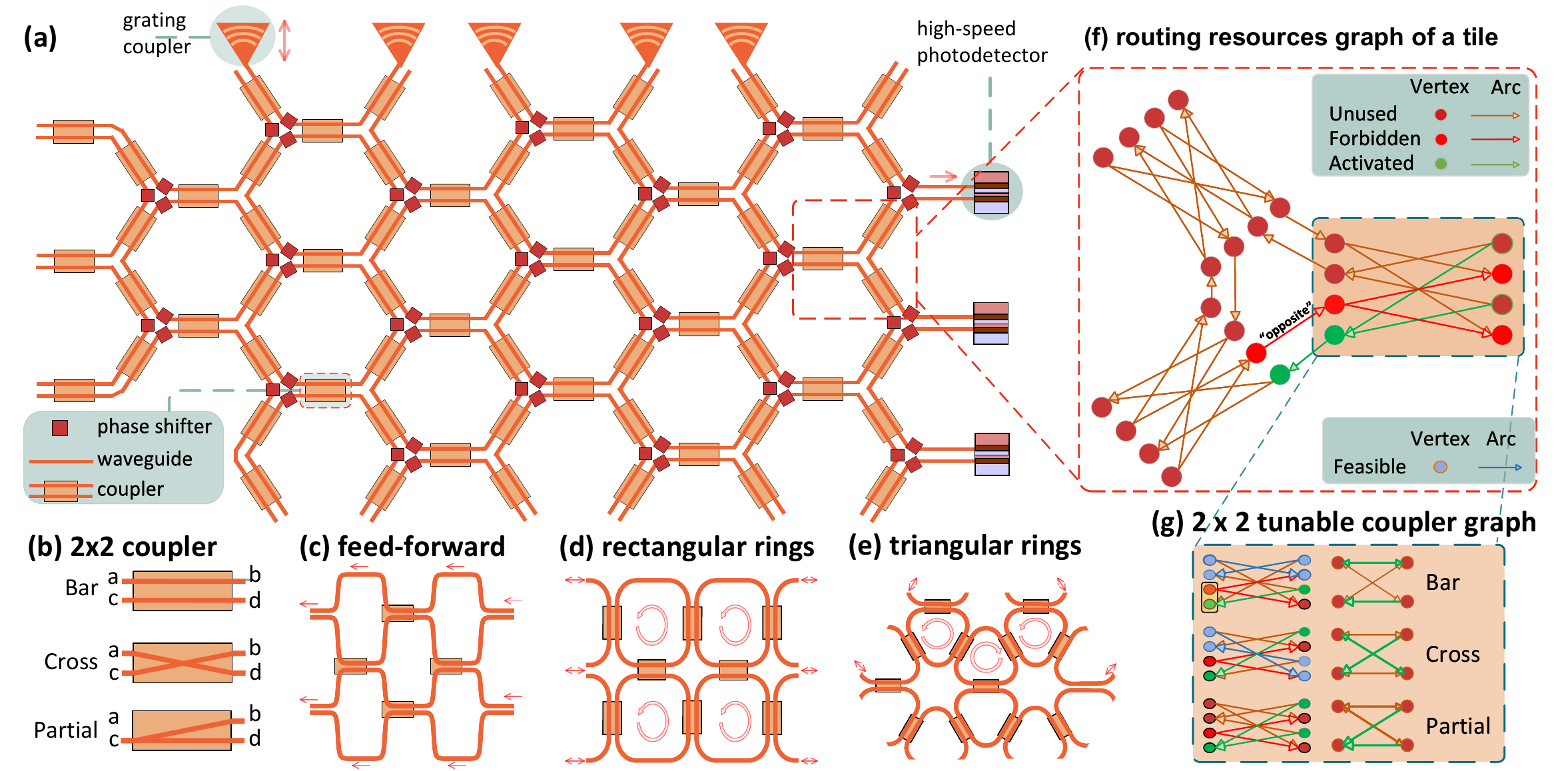}
    \caption{\small Hexagonal programmable photonic circuit with integrated phase shifters and tunable $2\times2$ couplers (a), showing coupler states (b), three other different architectures of programmable photonic circuits (c–e), and directed routing resources graph of tile and coupler (f, g).}
    \label{fig:model}
\end{figure*}
Silicon photonics and Photonic Integrated Circuits (PICs) have seen significant advancements in recent decades, offering compact, high-speed, and energy-efficient solutions for optical signal processing. Among these developments, programmable photonics has emerged as a promising approach~\cite{bogaerts_programmable_2020,perez2020multipurpose}, enabling reconfigurable circuits composed of tunable elements and functional building blocks. Unlike traditional PICs, where connectivity is fixed, programmable architectures allow dynamic configuration, introducing new challenges in circuit routing and resource allocation. Various routing algorithms have been proposed to address these challenges in programmable photonics, such as sequential routing~\cite{lopez_auto-routing_2020}, Aurora~\cite{kerchove_automated_2023}, and CRoute~\cite{wang2024croute}. However, existing solutions mainly focus on shortest-length routing while satisfying photonic constraints, with an emphasis on scalability and efficient resource usage.

In addition to resource utilization, another key aspect of routing in photonics circuits is the management of optical waveguide lengths. Unlike electrical interconnects, where wire length primarily affects resistance and capacitance, optical paths must be carefully controlled to maintain phase coherence and timing synchronization. The degree of length sensitivity varies depending on the application. In some cases, such as simple interconnects, precise length control is not critical. In contrast, for components like Mach-Zehnder Interferometers (MZIs), waveguide length directly affects phase delay and must be strictly managed. This distinction is particularly crucial in programmable PICs, where waveguide paths must be dynamically adjusted to achieve specific optical behaviour.

Length-matching techniques have long been studied in the printed circuit board (PCB) and analogue circuit design \cite{wong_pcb_2008, tan_yan_bsg-route_2008}, where ensuring signal arrival synchronization is essential to prevent timing mismatches. Similar constraints arise in photonics. Unlike electronic circuits, photonic circuits are susceptible to phase errors caused by path length mismatches. While phase shifters offer limited compensation—typically up to a single wavelength—larger mismatches from mesh-based routing cannot be corrected post-fabrication. Thus, precise length matching during routing is essential to ensure phase coherence and maintain the performance of programmable PICs. In traditional photonic circuits, physical design is mostly schematic-driven \cite{KORTHORST2023335}, and routing is often performed manually or optimized for specific layouts, and also global route planning \cite{ding2009oil,ding2009router}. However, in programmable PICs, the challenge is different: the interconnect structure is based on a regular tiling architecture, meaning that paths must be computed dynamically while satisfying strict length constraints. This introduces additional complexity, as conventional shortest-path routing is insufficient: Detours must be introduced to ensure that the path length (number of waveguide segments used) exactly meets the required value.
This makes length-matching in programmable PICs a unique problem that demands specialized routing algorithms. To the best of our knowledge, the length-matching routing problem in programmable photonics has not been previously addressed. The key distinctions between our approach and existing routing algorithms in this field are summarized in Table~\ref{tab:routing_compare}.
\begin{table}[ht]
\footnotesize
\setlength{\tabcolsep}{3pt} %
\centering
\caption{Comparison of routing algorithms for programmable photonic integrated circuits (PICs).}
\begin{tabular}{lccc}
\hline
\textbf{Router} & 
\makecell{\textbf{Length} \\ \textbf{Matching}} & 
\makecell{\textbf{Fan-out} \\ \textbf{Support}} & 
\makecell{\textbf{Global / Local} \\ \textbf{Router}} \\
\hline
Sequential Routing~\cite{lopez_auto-routing_2020} & \xmark & 2-pin only & Global \\
Aurora~\cite{kerchove_automated_2023}              & \xmark & 2-pin only & Global \\
C-Route~\cite{wang2024croute}                      & \xmark & \cmark  & Global \\
\textbf{This Work (H$\&$DS-LM)}                         & \cmark & \cmark & \textbf{Local Component} \\
\hline
\end{tabular}
\label{tab:routing_compare}
\end{table}
This work focuses on addressing this challenge by introducing an efficient search-based length-matching (LM) routing algorithm. The proposed approach extends best-first search strategies to enforce strict length constraints while exploring feasible detours within the regular tiling architecture of programmable PICs. Furthermore, a dual-stage LM variant is introduced as an extra option, allowing for initial shortest-path exploration followed by targeted detours (close to the target) to achieve precise length matching. The proposed methods provide a practical solution for length-matching in programmable photonic circuits, enabling more efficient and flexible circuit implementations.

\section{Preliminaries}

\subsection{Programmable Photonic Circuits}

We target programmable photonic integrated circuits (PICs) based on a hexagonal mesh of tunable $2 \times 2$ couplers and integrated phase shifters, as shown in Fig.~\ref{fig:model}(a). Each coupler supports three programmable states: bar, cross, and partial coupling (Fig.~\ref{fig:model}(b)), allowing flexible routing of light signals. 
Alternative architectures, such as feed-forward meshes and rectangular or triangular feedback loops, are illustrated in Fig.~\ref{fig:model}(c-e). In this paper, we limit our scope to problems based on the hexagonal architecture. For routing purposes, we construct a directed routing resource graph where each coupler and interconnect is represented by multiple vertices and arcs, as shown in Fig.~\ref{fig:model}(f-g). This abstraction captures the directionality of light propagation and enables enforcement of physical constraints such as no back-propagation and loop-free paths. Activated, forbidden, and feasible states are encoded in the graph to reflect current routing conditions.
\subsection{Problem Formulation}

We first clarify the routing requirements for programmable photonic circuits, including physical restrictions. The tunable coupler serves as a unique device within the routing process. Therefore, a customized model needs to be developed, and specific rules must be established to adhere to the physical restrictions.

The routing resources graph needs to obey the optical physical restriction, so we summarize the rules into the following two restrictions:

\begin{enumerate}
\item Acyclic: Avoiding loops (closed paths) in the circuit to prevent undesirable interference effects, in line with the routing requirements. 
\item Avoid Opposite: Each photonic connection or edge in the graph $G$ is inherently directional, allowing light propagation in only one specified direction along the waveguide during routing, as illustrated in Fig. \ref{fig:model}(f-g). A pair of two directional arcs represents each optical waveguide segment. Once an arc (shown in green) is utilized, its opposite arc (shown in red) becomes forbidden.
\end{enumerate}

\subsection{Length Constraints}
Except for the foundational restrictions introduced above, length-matching routing has an obvious extra constraint: the length. In typical photonic applications, such as optical true time delay lines (OTTDs), Mach-Zehnder interferometers (MZIs), and optical ring resonators (ORRs), the path length directly determines the device's optical functionality and phase relationship. For interference-based components such as the MZI shown in Fig.~\ref{fig:benchmark}(a), two arms must be length-matched, posing two length-matching routing problems, while other paths, such as input and output connections, are standard shortest-path problems. As illustrated in Fig.~\ref{fig:benchmark}(d), $L_1$ and $L_2$ denote the paths matched in length, and the rest are the shortest route. The length-constrained paths must also obey the previously mentioned acyclic and non-opposite restrictions.

\textbf{In programmable photonic meshes, the length constraint is especially strict:} each path consists of an integer number of identical segments, so the path length can only take discrete values, i.e., $L_i = N_i \cdot l_{\mathrm{seg}}$ for integer $N_i$. While general photonic circuits can tolerate small (sub-wavelength) path differences through phase tuning, the coarse granularity here means that even a one-segment mismatch ($|N_1 - N_2|=1$) results in a phase error many times greater than the optical wavelength, which cannot be compensated by on-chip phase shifters. Thus, exact segment-matching ($N_1 = N_2$) is a hard requirement for correct circuit operation.

\subsection{The Best-First Mechanism $BF$}
\label{sec: BF}
The best-first algorithm is a well-known search mechanism, designed to handle pathfinding tasks with improved efficiency by a numerical estimation involving a heuristic evaluation function $f(\cdot)$ \cite{Nilsson_a_star_1968, pearl_heuristics_1984}. 
The whole process is as follows:
{\small\begin{enumerate}[label={\arabic*.}]
    \item Add start node $s$ in the OPEN list (queue) which contains unexpanded nodes.
    \item Check if the OPEN list is empty. If it is, terminate the search as unsuccessful.
    \item Select and remove the node with the smallest \( f \) value from the OPEN list and move it to the CLOSED list.
    \item If this node is the goal, terminate the search successfully and reconstruct the path using pointers.
    \item Expand the current node, generate its successors, and assign each pointer back to the current node.
    \item For each successor $n'$:
    \begin{enumerate}[label={\arabic*.}]
        \item Calculate the evaluation function \( f \).
        \item If $n'$ is not in OPEN or CLOSED, add it to OPEN.
        \item If $n'$ is in OPEN or CLOSED and the new \( f \) value is lower, update its \( f \) value and move it back to OPEN if necessary.
    \end{enumerate}
    \item Return to step 2 and repeat until the goal is reached or the OPEN list is empty.
\end{enumerate}}

This iterative process ensures that the $BF^*$ algorithm systematically explores the most promising paths first, based on the heuristic and cost function defined, and adjusts its path strategy dynamically based on the results of successive node expansions. This mechanism serves as the foundation for the length-matching (LM) search strategies proposed in this work. In the following sections, we modify BF to incorporate exact path-length constraints, introducing H-LM and DS-LM, which adapt BF’s expansion strategy to meet the specific requirements of length-constrained routing.

\subsection{Heuristic Concepts for Length-Matching Routing}

To effectively guide the search, heuristic-based routing algorithms use an evaluation function \( f(\cdot) \) to prioritize node expansions. We introduce the essential definitions of heuristic functions relevant to this problem:

\begin{definition}[Admissible Heuristic]
A heuristic \( h(\cdot) \) is admissible if it never overestimates the true minimal cost \( h^*(\cdot) \) from the node \( n \) to the target:
\[
h(n) \leq h^*(n), \quad \forall n.
\]
\end{definition}

\begin{definition}[Consistent (Monotonic) Heuristic]
A heuristic \( h(\cdot) \) is consistent (monotonic) if, for each node \( n \) and its successor \( n' \), it satisfies the triangle inequality:
\begin{equation}
h(n) \leq c(n,n') + h(n'), \quad \forall (n,n')\in G.
\label{eq:consistency}
\end{equation}
where the $c(\cdot)$ is actual cost (length). Consistency implies admissibility.
\end{definition}

In the context of length-matching routing, maintaining strict heuristic monotonicity is challenging due to the necessity of node revisitation and detours. We briefly introduce the impact of these properties on the search:

\begin{itemize}
    \item \textbf{Completeness:} Ensures a path of exact length is found if it exists.
    \item \textbf{Consistency (Monotonicity):} Ensures systematic, efficient search, but can conflict with node revisitation requirements.
\end{itemize}

A detailed analysis of these properties for our heuristic estimator and LM algorithm is presented in the following sections.

\section{Length-matching Routing Algorithms}

In this section, we first introduce multiple algorithmic approaches designed for solving the length-matching (LM) routing problem, each characterized by distinct search strategies and complexity trade-offs. Then, we discuss node revisitation policies and the heuristic estimator specifically developed for hexagonal grid routing, both of which play a crucial role in enhancing search completeness and efficiency.

\subsection{Different LM Approaches}

In addition to our proposed routing strategies, we reference a previously established method, the $A^*$-based length-matching ($A^*$-LM) algorithm from the LEMAR router~\cite{yao2012lemar}. This method performs a constrained search over a tile-based routing graph to identify paths whose lengths fall within a specified tolerance range. While it employs an $A^*$-style expansion strategy using a cost-based priority queue, its wavefront progression mechanism restricts detour flexibility during early search stages. As a result, the algorithm primarily explores shortest paths toward the target and postpones any path-length adjustment until the destination is reached. This forces a greedy backtracking process to extend the path, rendering much of the earlier wavefront search ineffective for exact-length matching. Furthermore, the lack of heuristic detour guidance and the absence of a visited-node mechanism lead to excessive redundant node expansion, increasing search overhead, particularly under strict length-matching constraints.
Subsequently, we introduce a simpler greedy LM search as a baseline to illustrate similar unconstrained exploration behaviors. This sets a clear foundation for contrasting our enhanced heuristic-driven approaches. Finally, we present two advanced algorithmic strategies, each uniquely defined by its evaluation methodology and revisitation policy, specifically tailored to address the photonic LM routing challenge more effectively.

\subsubsection{Greedy LM Search (No Heuristic)}
The first approach is a straightforward Best-First greedy method. It prioritizes nodes using a simple cost function without heuristic guidance:
\begin{equation}
f(n) = L - g(n),
\label{eq:greed-lm}
\end{equation}
where $g(n)$ is the accumulated cost (length) from the source to node $n$, and $L$ is the desired length. Since the search prioritizes the largest $g(n)$ (smallest remaining length), this approach essentially behaves like a depth-first exploration. It allows node revisitation, which guarantees completeness but may lead to exponential search complexity in the worst case. In this work, we include this greedy method as a baseline for comparison.

\subsubsection{Heuristic-Guided LM (H-LM) Search}
The second approach incorporates heuristic guidance into the LM problem. The evaluation function combines a heuristic estimator $h(n)$, yielding:
\begin{equation}
f(n) = L - g(n) - h(n),
\end{equation}
where $h(n)$ is admissible and consistent, providing an underestimate of the distance from node $n$ to the target. The search space is illustrated in Fig.~\ref{fig:maze_astar_a}, depicting the path from the source to the target. The orange blocks represent the source \textit{S} and target \textit{T}, while the colored blocks stand for the searched nodes with evaluated cost $f_n$. In this case, within a Manhattan grid, the heuristic perfectly estimates the minimum possible cost, if there are no barriers or congestion. As a result, the search proceeds optimally without any node revisitation, representing the ideal performance of this algorithm. From the figure, we can clearly distinguish two steps: first, finding detours based on gradient descent cost; and second, once the detours are enough, the algorithm proceeds in a manner similar to A*. Although heuristic guidance can significantly improve computational efficiency by directing the search toward promising nodes, it may sometimes be misled by an underestimated heuristic, leading to the selection of infeasible paths. Once sufficient detours have been explored, the shortest path estimated by the heuristic may no longer exist in the actual search space. Ultimately, the algorithm may degenerate into a greedy exploration strategy. 

The pseudo-code is presented in Algorithm~\ref{alg:contrained_A*}, mainly relying on the Best-First searching mechanism. $G$ is the graph, $q$ is the priority queue that stores the OPEN nodes, and $s$ is the source node at the beginning. $u$ represents the current node being popped out, $v$ is the successor of $u$, $pt$ is the stored path, $g$ is the distance from source to current node and $e$ is the edge length. Line 11 defines the evaluation function, and line 12 implements a pruning strategy, which is further detailed in Section~\ref{subsec:prun}.

\begin{algorithm}[ht]
\caption{Heuristic-Guided LM Routing Algorithm}
\footnotesize
\label{alg:contrained_A*}
\begin{algorithmic}[1]
\Procedure{$H\text{-}LM$}{$G, s, t, L$}
    \State $q \gets [(0, s, [], 0)]$  \Comment{OPEN list}
    \While{$q$}
        \State $(f, u, pt, g) \gets \Call{HeapPop}{q}$
        \State $pt \gets pt + [u]$
        \If{$u = t$ \textbf{and} $g = L$}
            \State \Return $pt$
        \Else
            \ForAll{$(v, e) \in G[u]$}  \Comment{Node expansion}
                \State $g_v \gets g_n + e(u,v)$
               \State $\boldsymbol{f_v \gets L - (g_v + h(v,t))}$ \Comment{Evaluation}
                \If{$f_v \geq 0$} 
                \State $\Call{HeapPush}{q, (f_v, v, pt, g_v)}$ 
                \EndIf
            \EndFor
        \EndIf

    \EndWhile
    \State \Return $None$
\EndProcedure
\end{algorithmic}
\end{algorithm}

\subsubsection{Dual-Stage LM (DS-LM) Routing}
The third approach, the Dual-Stage LM routing algorithm, balances the above two strategies by separating the search into two stages:
\begin{itemize}
    \item \textbf{Stage 1: Shortest Path Search}:  (lines 3–10 in Algorithm~\ref{alg:dual_stage_LM}) Finds the shortest route from source to destination quickly using standard $A^*$ heuristic guidance. If the length of this shortest path matches the desired length $L$, the search terminates immediately.
    \item \textbf{Stage 2: Detour Search}:  (lines 12–17 in Algorithm~\ref{alg:dual_stage_LM}) If the shortest path is shorter than $L$, the algorithm switches to the LM evaluation function ($f(n)=L-g(n)-h(n)$) and systematically explores detours. Node revisitation is allowed to ensure that the algorithm exhaustively explores potential paths, thus maintaining completeness at the expense of potential exponential complexity in worst-case scenarios.
\end{itemize}

This dual-stage method postpones the detour step to the end, near the target, potentially yielding different results, as illustrated in Fig.~\ref{fig:maze_astar_b}. The darkness of nodes refers to current length $g_n$ here. If the source node has a large fan-out and thus insufficient available edges for detouring, delaying the detour to a later stage becomes an ideal strategy. But this method requires two stages and also needs to reinitialize the priority queue in between for transferring new evaluation functions. This reprocessing step introduces additional computational overhead, making DS-LM potentially more time-consuming than other approaches.

\begin{figure}[ht]
    \centering
    \begin{subfigure}[b]{0.48\linewidth}
        \centering
        \includegraphics[width=\linewidth]{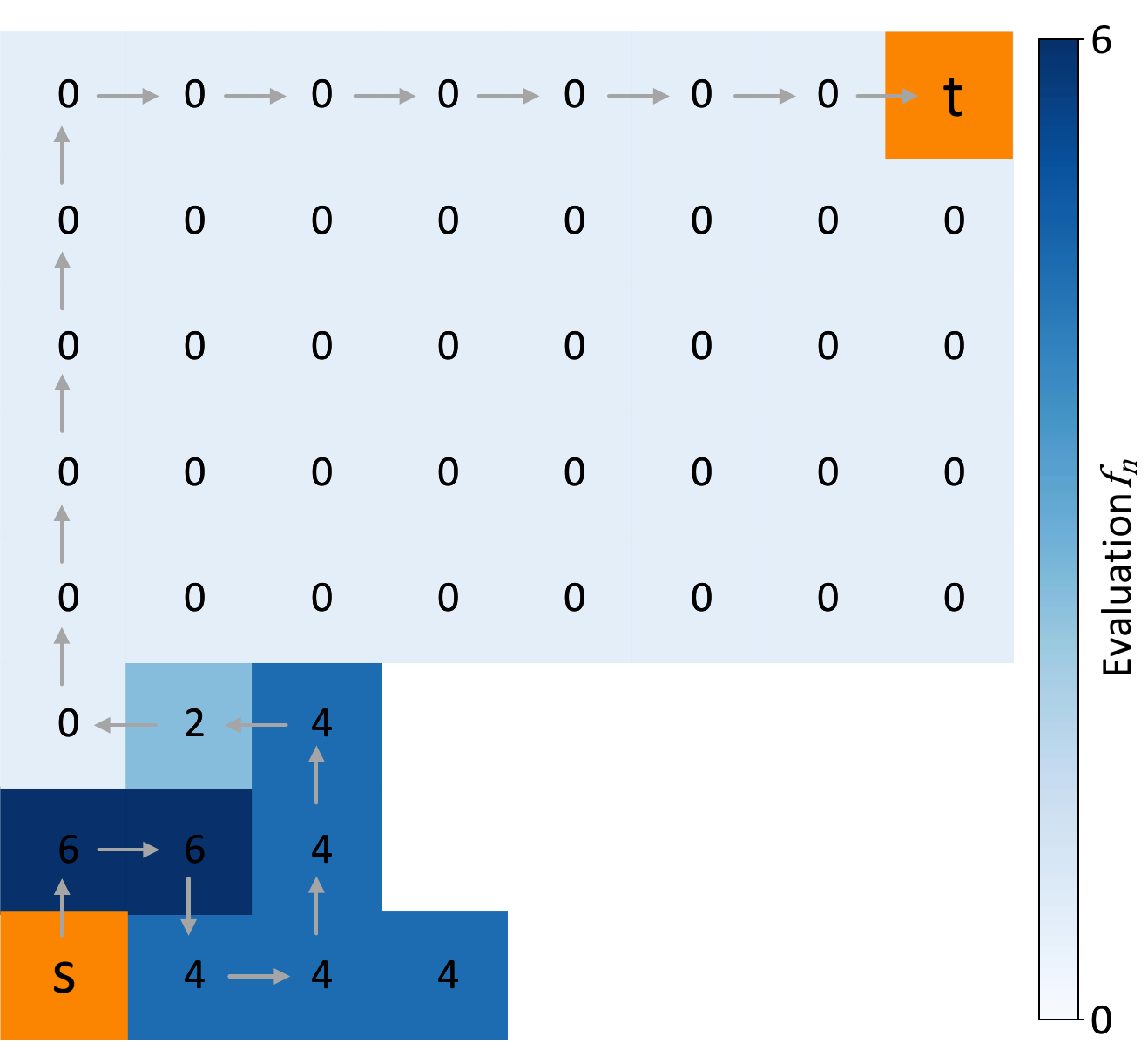}
        \caption{H-LM (detour first)}
        \label{fig:maze_astar_a}
    \end{subfigure}
    \hfill
    \begin{subfigure}[b]{0.48\linewidth}
        \centering
        \includegraphics[width=\linewidth]{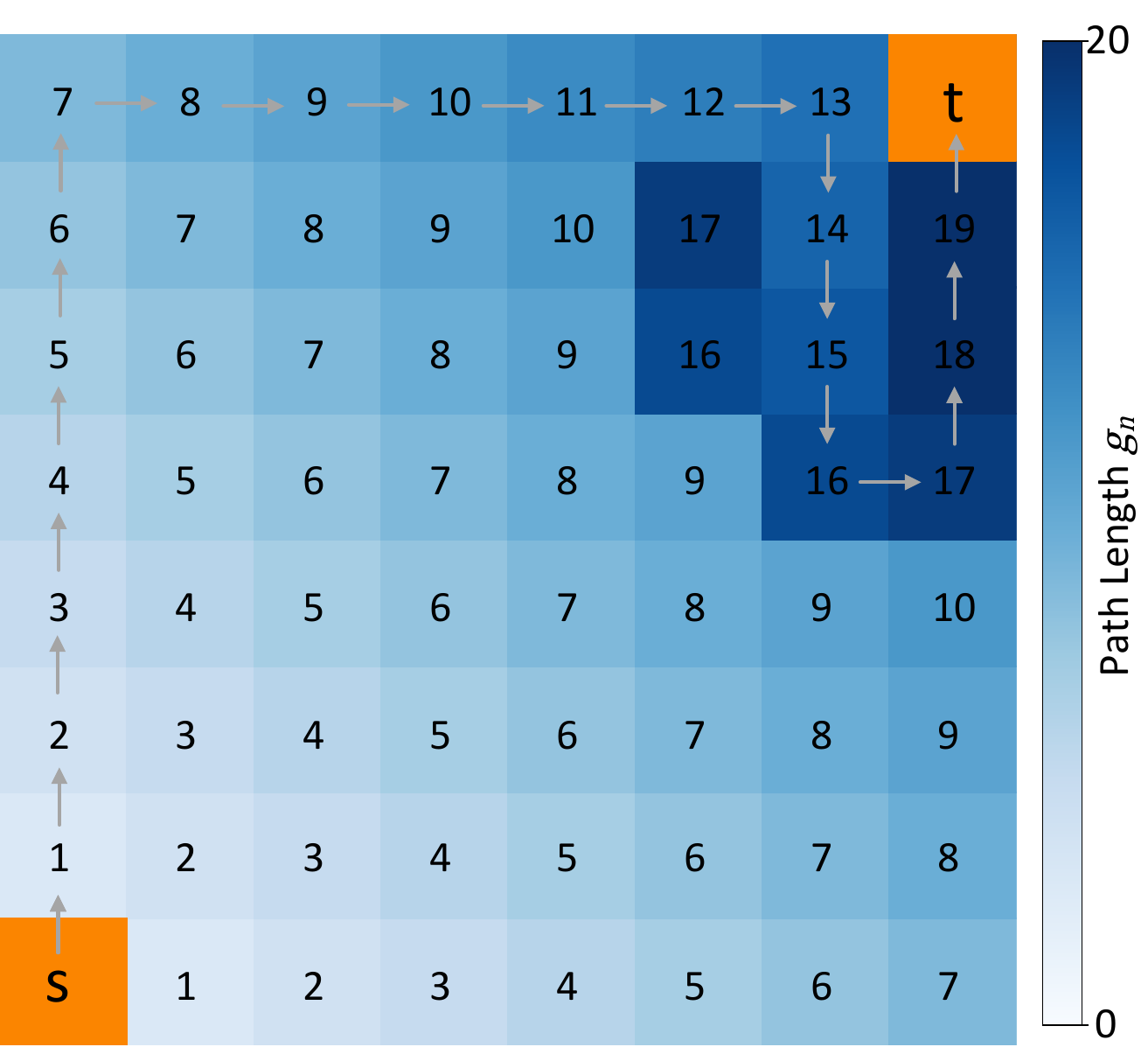}
        \caption{DS-LM (detour next)}
        \label{fig:maze_astar_b}
    \end{subfigure}
    \caption{\small An example of maze routing solved by H-LM and DS-LM with the desired exact length $L=20$ between source $s$ and target $t$ in the Manhattan grid\protect\footnotemark{}.}
\end{figure}

\begin{algorithm}[ht]
\footnotesize
\caption{DS-LM Routing}
\label{alg:dual_stage_LM}

\begin{algorithmic}[1]
\Procedure{$DS\text{-}LM$}{$G, s, t, L$}
    \State \textbf{Init:} $q \gets [(0, s, [], 0)]$, $vis \gets \emptyset$, $scost \gets \infty$
    \While{$q$}  \Comment{\textbf{\textsc{Stage 1: Shortest Path Search}}}
        \State $(g_u, u, pt, \_) \gets \Call{HeapPop}{q}$
        \If{$u \in vis$ \textbf{and} $g_u > vis[u]$} \textbf{continue} \EndIf
        \State $vis[u] \gets g_u$, $pt \gets pt + [u]$
        \If{$u = t$} 
            \State $scost \gets g_u$
            \textbf{break}
        \EndIf
        \ForAll{$(v, e) \in G[u]$} 
            \State $g_v = g_u + e(u,v)$
            \State $\Call{HeapPush}{q, (g_v + h(v,t), v, pt, g_v)}$
        \EndFor
    \EndWhile
    \If{$scost \geq L$} \Return $pt$ \EndIf
        \State \textbf{Init:} $q \gets \{(L - g_u, u, pt, g_u) \text{ for } (f, u, pt, g) \in q \}$
        \Statex \hspace{1em} \Comment{\textbf{\textsc{Switch Cost Function}}}

    \While{$q$}   \Comment{\textbf{\textsc{Stage 2: Detour Search}}}
        \State $(f, u, pt, g_u) \gets \Call{HeapPop}{q}$
        \State $pt \gets pt + [u]$
  
        \If{$u = t$ \textbf{and} $g_u = L$} \Return $pt$ \EndIf
        \ForAll{$(v, e) \in G[u]$}
            \State $g_v = g_u + e(u,v)$
            \State $\Call{HeapPush}{q, (L - g_v - h(v,t), v, pt, g_v)}$
        \EndFor
    \EndWhile
    \State \Return $None$
\EndProcedure
\end{algorithmic}
\end{algorithm}

\subsection{Node Expansion and Revisitation}
In the proposed LM routing algorithm, nodes must be allowed unrestricted revisitation to ensure completeness, due to the necessity of detours for achieving the exact path length \(L\). Later in Section~\ref{sec:formal_prop}, we prove that the evaluation function \(f(n)=L-g(n)-h(n)\) is theoretically non-increasing but lacks practical monotonicity, making revisitation essential. Therefore, restrictive node revisitation policies could prematurely eliminate viable solutions. Instead, revisitation should be permitted freely (no CLOSED restriction in BF mechanism), with pruning solely determined by an admissible heuristic evaluation criterion \(f(n)\geq 0\) introduced later in Section~\ref{subsec:prun}.

\footnotetext{In this Manhattan grid, the heuristic estimate is Manhattan distance $h(\cdot) = \delta x + \delta y = \left| x_1 - x_2 \right| + \left| y_1 - y_2 \right|$, which equals $h^*(\cdot)$ and is also admissible.}

\subsection{A Heuristic Estimator on Hexagonal Grid}
\label{estimator}
The hexagonal programmable PIC circuits are physically hexagonal, but the interconnect structure is not. The \textbf{`switching blocks'} are couplers placed at the borders between two hexagons, rather than at the corners where three hexagons meet, as shown in Fig.~\ref{fig:model}. Consequently, the corner locations (phase shifters) do not have reconfigurable interconnect capabilities. Therefore, the graph representation must be reconsidered to derive a more informed heuristic estimation (close to the optimized shortest length), improving upon the method proposed in \cite{wang2024croute}.

\begin{figure}[ht]
    \centering
    \includegraphics[width=.7\linewidth]{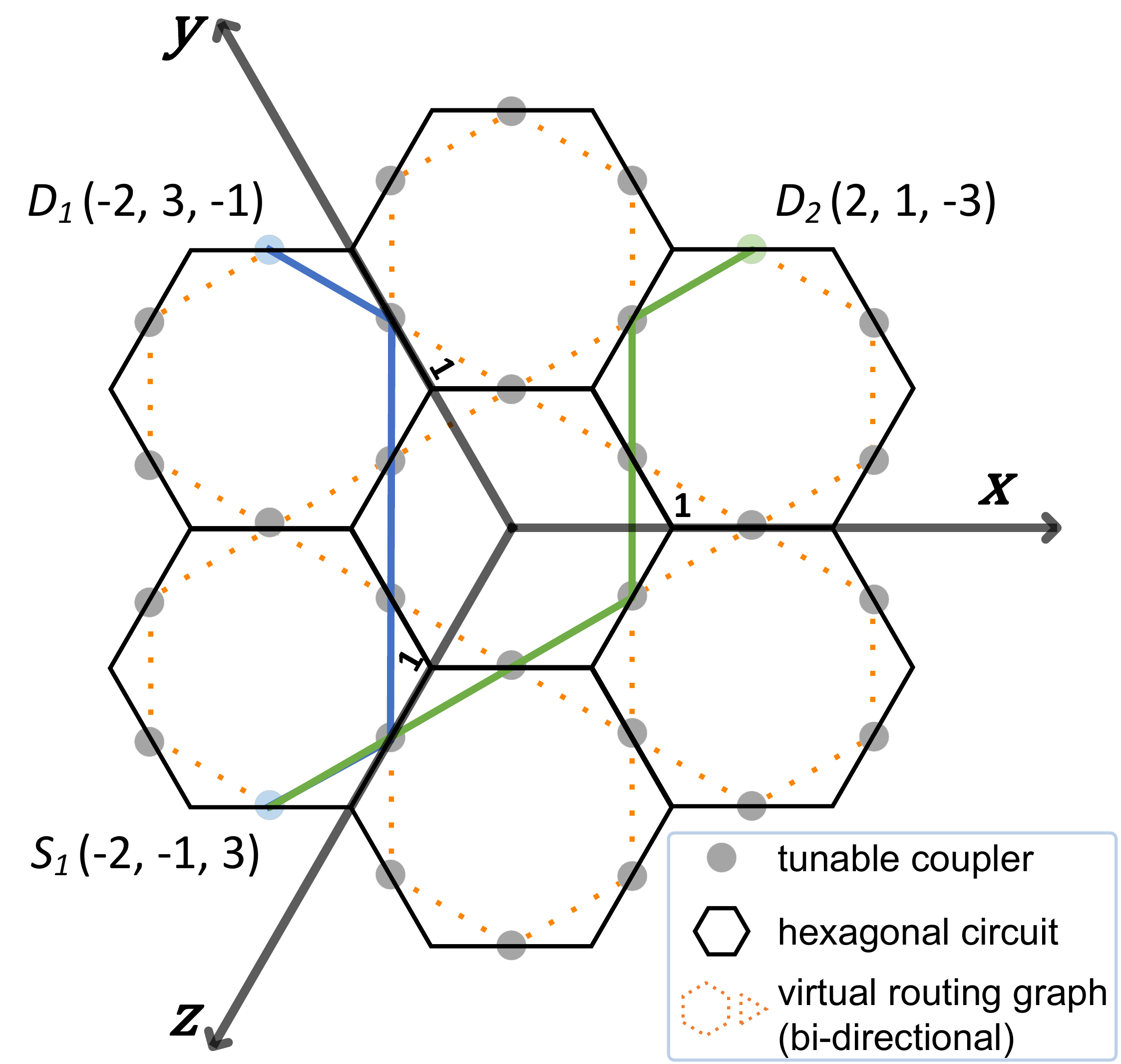}
    \caption{Look-ahead distance calculation of the heuristic estimator for a hexagonal graph.}
    
    \label{fig:estimation}
\end{figure}

As shown in Fig.~\ref{fig:estimation}, we consider only the inter-coupler waveguide length while ignoring the intra-coupler waveguide in the graph. This simplification is valid since any routed path must traverse both inter- and intra-waveguides as a pair. Thus, each coupler is represented as a node, the actual routing resources graph (RRG) is the orange mesh, tiling with hexagons and triangles. So we use the following equation to calculate the look-ahead distance for both the H-LM and DS-LM algorithms:

\begin{equation}
h(n,t) = 
    \begin{cases}
    \max\limits_{i \in \{x, y, z\}} |S_i - D_i| + 1, & \text{if } S_i = D_i \\ &\hspace*{\fill}\text{ for any } i \in \{x, y, z\}, \\
    \max\limits_{i \in \{x, y, z\}} |S_i - D_i|, & \text{otherwise}.
    \end{cases}
    \label{eq:heuristic}
\end{equation}

The distance between $S_1$ and $D_1$ is thus $3 - (-1) + 1$, as they have the same coordinate on the $x$-axis, requiring one detour (e.g. the distance between $S_1$ and $D_1$). In all other cases (e.g. when $S_1$ and $D_2$ do not share the same coordinate on any axis), the maximum absolute coordinate difference determines the shortest distance. 
As the built graph for routing is a direction graph, the estimation is based on the bidirectional graph. The distance is not always reachable for the shortest path. Also, considering occupation and congestion, this estimator is less than the optimal one $h(n) \leq h^*(n)$. Thus, this heuristic estimator is admissible (never overestimates), providing reliable guidance for the length-matching routing algorithms.

Before formally proving consistency, we state a supporting lemma:

\begin{lemma}\label{lem:triangle}
For any nodes $a,b,t$ in the hexagonal RRG, the heuristic (Eq.~\ref{eq:heuristic}) satisfies the triangle inequality:
\begin{equation}
    h(a,t) \leq c(a,b) + h(b,t).
\end{equation}
\end{lemma}

\begin{proof}
Consider arbitrary nodes $a$, $b$, and target $t$ with coordinates $a_i,b_i,t_i$ for $i\in\{x,y,z\}$. By definition of $h(\cdot)$ from Eq.~\ref{eq:heuristic}:
\begin{align}
h(a,t) &= \max_{i}|a_i - t_i| \nonumber\\
&\leq \max_{i}\left(|a_i - b_i| + |b_i - t_i|\right) \quad\quad(\text{triangle inequality}) \nonumber
\\
&\leq \max_{i}|a_i - b_i|+\max_{i}|b_i - t_i| \nonumber\\
&= c(a,b)+h(b,t), \label{eq:heurstic_admissible}
\end{align}
where $c(a,b)$ denotes the actual minimum distance (cost) between $a$ and $b$ on the grid. Thus, the heuristic satisfies the triangle inequality and is consistent.
\end{proof}

\section{Formal Properties of LM Algorithms}
\label{sec:formal_prop}

\subsection{Greedy-LM Behaviour}
\label{subsec:greedy-dfs}
The first Greedy-LM was proposed with the evaluation function shown in Eq.~\ref{eq:greed-lm} without information from the heuristic.

This evaluation function strictly decreases along the paths:
\[
f(n') = L - g(n') = L - [g(n) + c(n,n')] < f(n),
\]
indicating that the algorithm behaves like a depth-first search (DFS), fully exploring deeper paths (higher cost) before searching others. Thus, this approach is effectively depth-first due to its greedy exploration strategy.

\subsection{Monotonicity in Length-Matching Routing (H-LM $\&$ DS-LM)}

In length-matching (LM) routing, the desired path length $L$ is restricted to be equal to or greater than the shortest possible path length from source to destination. Hence, the algorithm inherently requires detours, permitting node revisitation. This revisitation fundamentally disrupts the monotonic property typically enjoyed by heuristic search.

Consider the evaluation function:
\[
f(n) = L - g(n) - h(n).
\]
For a node $n$ and its successor node $n'$, we have:
\begin{align}
f(n') &= L - g(n') - h(n') \nonumber\\
      &= L - [g(n) + c(n,n')] - h(n') \nonumber\\
      &\leq L - g(n) - h(n) = f(n), \quad \text{(by Eq.~\ref{eq:consistency})}
\end{align}

where we demonstrate mathematically that $f(\cdot)$ is \textbf{non-increasing} along any path under a consistent heuristic.
However, a fundamental conflict arises here. Maintaining practical monotonicity would require expanding the nodes starting from the largest \( f(n) \) first since \( f(\cdot) \) is non-increasing. However, by heuristic search principles, the algorithm should prioritize expanding nodes with the smallest \( f(n) \), as these represent the most promising partial solutions (close to the target length). If the algorithm strictly followed the theoretical monotonic order, expanding the nodes from the largest \( f(n) \) first, it would consistently expand the least promising nodes, drastically reducing search efficiency. Furthermore, since a revisitation is necessary (each node may represent distinct paths through various detours), the search space could expand exponentially, resulting in prohibitive computational complexity.\\
On the other hand, if the algorithm prioritizes nodes with the smallest (\textbf{best}) \( f(n) \), as heuristic guidance naturally suggests, it inherently violates practical monotonicity.
Such a heuristic-driven expansion strategy, while efficient, sacrifices monotonicity, potentially losing completeness if pruning conditions are not carefully chosen.

\subsection{Pruning}
\label{subsec:prun}

To ensure completeness, the LM algorithms (H-LM and DS-LM) employ a greedy revisitation strategy with the pruning criterion:
\[
f(n) = L - g(n) - h(n) \geq 0.
\]

As our heuristic estimator $h(\cdot)$ adheres to the admissibility $h(n) \leq h^*(n)$ proven by Eq.~\ref{eq:heurstic_admissible}, 
completeness is preserved under this pruning rule.

\begin{proof}
Consider a node $n$ lying on an optimal path with exact length \(L\). Due to heuristic admissibility, we have:
\[
g(n) + h(n)\leq g(n) + h^*(n) = L.
\]

Thus, the evaluation function satisfies:
\[
f(n)=L - g(n) - h(n)\geq0
\]
only when the path length exactly matches or exceeds the required length \(L\). Therefore, no path potentially leading to an exact length solution is prematurely pruned, ensuring completeness.
\end{proof}



\begin{figure*}
    \centering
    \includegraphics[width=.85\linewidth]{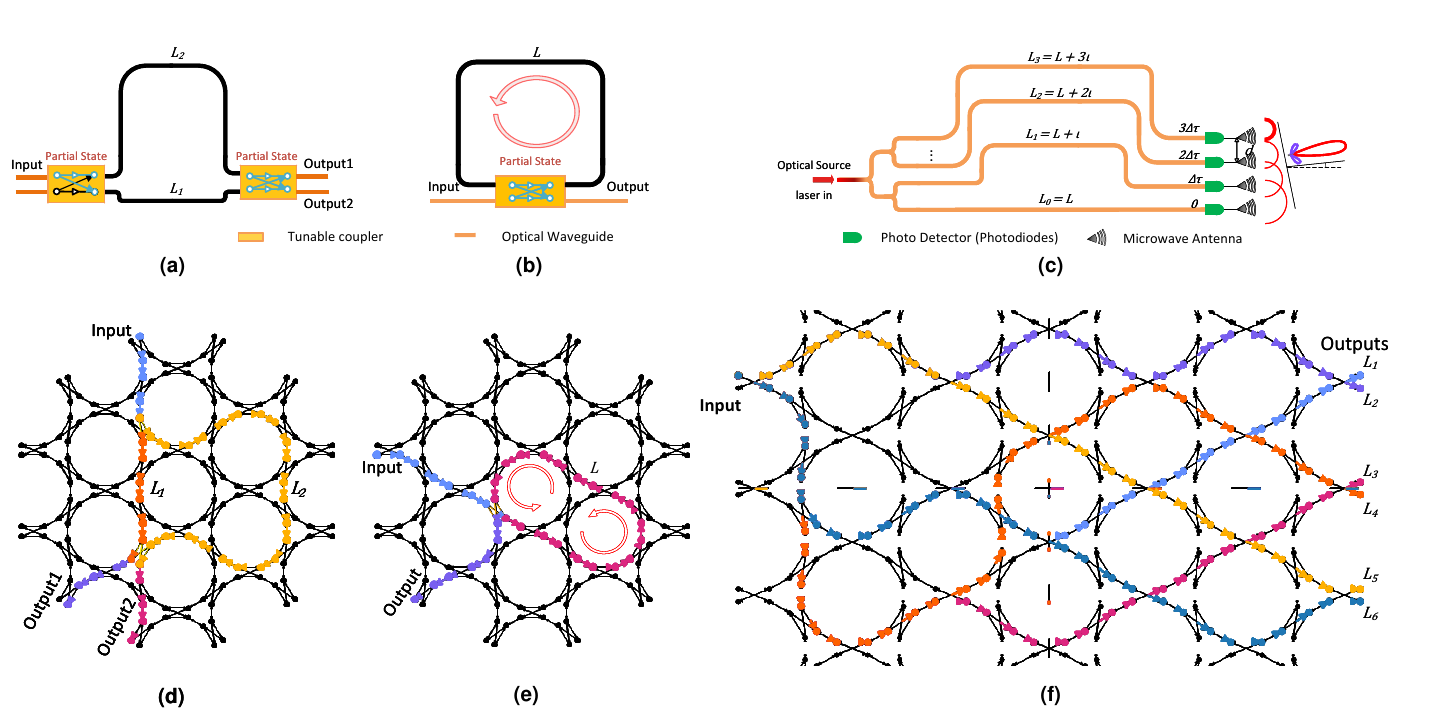}
    \caption{Benchmarks designed to reflect three representative photonic applications: (a–c) illustrate the schematic diagrams of a Mach–Zehnder Interferometer (MZI), an Optical Ring Resonator (ORR), and an Optical True Time Delay (OTTD) line for RF beamforming, respectively. (d–f) show the corresponding routing results on the hexagonal mesh using the H-LM algorithm. Colored paths represent routes that satisfy specific length constraints $L_i$.}

    \label{fig:benchmark}
\end{figure*}

\begin{figure}[h]
    \centering
    \includegraphics[width=.6\linewidth]{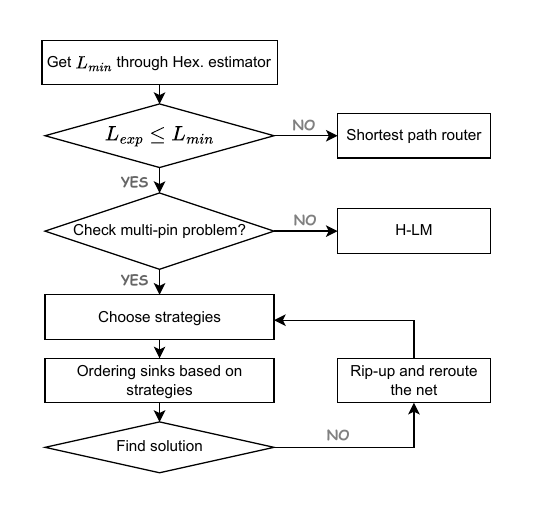}
    \caption{ Flowchart of the router with additional mechanisms for addressing multi-pin length-matching routing problems.}
    \label{fig:flowchart}
\end{figure}

\section{Additional Enhancement Mechanisms for Routing Multi-pin Nets}

It is important to note that the completeness guarantees discussed earlier are established specifically for two-pin nets. In contrast, for multi-pin nets, the routing process becomes inherently heuristic, as connections must be established sequentially, pair by pair. This sequential nature can easily lead to suboptimal local decisions, which may prevent the discovery of a globally feasible solution, thus invalidating the completeness guarantees in this broader context. To address this limitation and improve both efficiency and solution quality, we propose several enhancements to the three BF-based routing algorithms. These include evaluating the necessity of a length-matching router using the previously introduced heuristic estimator, prioritizing among different length-constrained connections, and conditionally changing the routing strategy. Since all routing strategies use length as the sole cost metric, the traditional PathFinder~\cite{mcmurchie_pathfinder_1995} mechanism, which relies on iterative dynamic cost adjustment to resolve congestion and violation, is no longer applicable. These enhancements are designed to work in concert, as illustrated in Fig.~\ref{fig:flowchart}, to improve routing robustness and solution completeness in the presence of complex multi-pin topologies.

\subsubsection{Ordering}
For multi-pin net routing problems with varying length constraints, such as the multicasting net requirements in beamforming, sink ordering plays a crucial role in identifying feasible solutions. In the H-LM algorithm, routing may introduce detours starting from the source. The source region will likely experience congestion, violation and thermal effects. Therefore, sinks with smaller detour margins (that is, minimal ${L_{\text{exp}} - L_{\text{min}}}$) should have higher priority in the order process. Conversely, sinks with larger detour margins can be deferred to later stages.
In contrast, for the DS-LM algorithm, detours tend to be introduced closer to the sink side rather than near the source. As a result, the impact of ordering is generally less critical when the sinks are well separated. However, if multiple sinks are located in close proximity, a similar issue as in H-LM may arise—early routes could block favorable paths for later ones. In such cases, prioritizing sinks with smaller detour margins remains beneficial.

\subsubsection{Rip-up and Re-route}
Although reordering the connections within a net in H-LM can improve routability, its success rate remains lower than DS-LM due to the latter's detour-first characteristic. To enhance general robustness, we adopt a two-step approach: the router first attempts H-LM; if it fails to find a feasible solution for the entire net, the routing is ripped up and retried using DS-LM. While not dynamically adaptive on a per-connection basis, this staged fallback mechanism balances efficiency and completeness.

\section{Measurements and Results Analysis}



\subsection{2-pin Benchmarks}

As shown in Fig.~\ref{fig:benchmark}(a-c), three length-matching application scenarios can be solved by the LM algorithm. Fig.~\ref{fig:benchmark}(d-e) are corresponding solutions on the hexagonal directional graph $G$. Based on these three application scenarios, we generated a group of benchmarks to test the performance of the LM algorithms. For the MZI ($B_1$) benchmarks, the source and target are separated by approximately three arc distances. In the optical ring resonator (OOR) ($B_2$) benchmarks, the source and target node pairs are located within the same coupler (since, in our definition, intra-coupler waveguides are ignored). Additionally, the distance between the source and target in the optical true time delay lines (OTTD) ($B_3$) benchmarks is eight arcs. Fig.~\ref{fig:line_avg_push} shows the average number of pushed nodes under increasing length constraints on one single benchmark problem as an example. The LEMAR algorithm based on $A^*$ consistently performs the worst, due to its rigid wavefront expansion and lack of visited node tracking, which leads to inefficient detour handling and redundant exploration. Therefore, subsequent comparisons focus on the more effective Greedy-LM, H-LM, and DS-LM strategies.
The results in Table~\ref{tab:benchmark_summary} and Fig.~\ref{fig:bar_benchmark_dual} illustrate the different search behaviours of the three algorithms: Greedy-LM, H-LM, and DS-LM. The greedy algorithm explores the largest number of nodes before finding a solution, as it systematically expands all possible paths without heuristic guidance. In contrast, H-LM, leveraging heuristic estimates, searches the fewest nodes, making it the most efficient in terms of space complexity. DS-LM lies between the two, as it initially follows heuristic guidance but later transitions into a detour-based length-matching phase.

\begin{figure}[ht]
\begin{tikzpicture}
\begin{axis}[
    width=8cm,
    height=3.5cm,
    xlabel={Expected Length},
    ylabel={Avg Pushed Nodes},
    xlabel near ticks,
    ylabel near ticks,
    legend style={
        at={(0.5,1.02)},
        anchor=south,
        legend columns=4,
        draw=none,
        font=\tiny,
        /tikz/every even column/.append style={column sep=1pt}
    },
    grid=major,
    xmin=8, xmax=26,
    ymin=0, ymax=2500,
    xtick={8,10,...,26},
    ytick={0,500,...,2500},
    tick label style={font=\tiny},
    label style={font=\tiny},
    line width=1pt,
    mark size=2.5pt,
]

\addplot+[
    color=purple!50,
    mark=square*,
    mark options={fill=purple!50, draw=purple!50}
] coordinates {
    (8,75) (10,195) (12,367) (14,631) (16,951)
    (18,1239) (20,1530) (22,1779) (24,2034) (26,2281)
};
\addlegendentry{LEMAR-like}

\addplot+[
    color=greencustom!70,
    mark=diamond*,
    mark options={fill=greencustom!70, draw=greencustom!70}
] coordinates {
    (8,138) (10,171) (12,192) (14,120) (16,96)
    (18,135) (20,591) (22,150) (24,135) (26,417)
};
\addlegendentry{Greedy-LM}

\addplot+[
    color=blue!40,
    mark=triangle*,
    mark options={fill=blue!40, draw=blue!40}
] coordinates {
    (8,111) (10,36) (12,64) (14,74) (16,85)
    (18,91) (20,98) (22,109) (24,115) (26,122)
};
\addlegendentry{H-LM}

\addplot+[
    color=orange!60,
    mark=*,
    mark options={fill=orange!60, draw=orange!60}
] coordinates {
    (8,91) (10,117) (12,140) (14,151) (16,156)
    (18,124) (20,410) (22,164) (24,303) (26,199)
};
\addlegendentry{DS-LM}

\end{axis}
\end{tikzpicture}
\vspace{-1mm}
\caption{Comparison of average search space (pushed nodes) across increasing expected path lengths for different routing algorithms, with a fixed source–sink distance of 8.}
\label{fig:line_avg_push}
\end{figure}
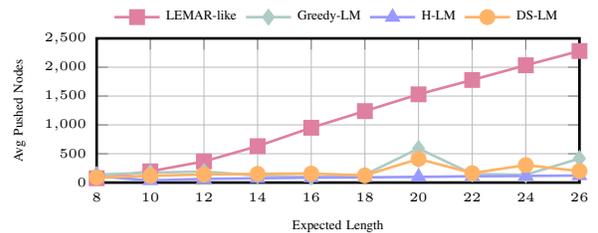

\begin{table*}[ht]
\small
\caption{ Performance Comparison of LEMAR-like, Greedy-LM, H-LM, and DS-LM Across Different Circuit Benchmarks.}
\label{tab:benchmark_summary}
\centering
\begin{tabular}{lcccccccc}
\hline
\multirow{2}{*}{\textbf{Benchmark}} 
& \multicolumn{2}{c}{\textbf{LEMAR-like\cite{yao2012lemar}}} 
& \multicolumn{2}{c}{\textbf{Greedy-LM}} 
& \multicolumn{2}{c}{\textbf{H-LM}} 
& \multicolumn{2}{c}{\textbf{DS-LM}} \\
\cline{2-3} \cline{4-5} \cline{6-7} \cline{8-9}
& Runtime (s) & Total Push 
& Runtime (s) & Total Push 
& Runtime (s) & Total Push 
& Runtime (s) & Total Push \\
\hline
$B_1$ (MZI)  & --- & ---  & 0.60  & 48356 & \textbf{0.10}  & \textbf{6859}  & 0.70  & 37405 \\
$B_2$ (ORR)  & ---     & ---   & 1.02  & 70950 & \textbf{0.78}  & 51372 & 0.96  & \textbf{49919} \\
$B_3$ (OTTD) & 0.0926     & 12086   & 0.009 & 1728  & \textbf{0.007} & \textbf{883}   & 0.016 & 1656  \\
\hline
\end{tabular}

\vspace{0.5ex}
\begin{minipage}{0.9\textwidth}
\raggedright
\tiny \textsuperscript{*} Runtime measured on AMD Ryzen 7 7800X 12-core processor with 128~GB RAM.
\end{minipage}
\end{table*}

However, DS-LM tends to be the most time-consuming strategy in $B_1$ and $B_2$ benchmarks due to its two-stage process. Specifically, as seen in Algorithm~\ref{alg:dual_stage_LM}, line 11, once the first stage terminates, the entire queue (containing all previously searched nodes) must be reprocessed with an updated evaluation function. This transition step incurs additional computational overhead, leading to increased runtime despite a reduced search space.
The effectiveness of heuristic-based search varies significantly across different circuit layouts. In long straight-line structures such as OTTDs ($B_3$), the heuristic function is highly informative, leading to significant reductions in both search space and runtime. This explains why H-LM performs exceptionally well in $B_3$. However, in source-target proximity scenarios like ORRs ($B_2$), where frequent detours are required between closely spaced nodes, even if $h(n)$ matches the optimal shortest distance $h^*(n)$, its guidance remains limited. As a result, all three algorithms exhibit more greedy behaviour, with reduced performance differences.

\begin{figure}[ht]
\centering
\begin{tikzpicture}
\begin{groupplot}[
  group style={
    group size=2 by 1,
    horizontal sep=0.7cm,
  },
  width=.6\linewidth,
  height=3.5cm,
  ybar,
  enlarge x limits=0.2,
  symbolic x coords={$B_1$, $B_2$, $B_3$},
  xtick=data,
  x tick label style={rotate=0, font=\tiny},
  tick label style={font=\tiny},
  ylabel style={font=\tiny, yshift=-6.3pt},
  title style={font=\tiny},
  legend style={
      at={(-0.65, 1.35)},
      anchor=north,
      legend columns=4,
      /tikz/every even column/.append style={column sep=1pt},
      draw=none,
      font=\tiny
    }
]

\nextgroupplot[
  ylabel={Runtime Ratio},
  ymin=0, ymax=2.0,
  legend style={at={(0.5,1)}, anchor=south, legend columns=3},
]

\addplot+[bar shift=-2.5pt, bar width=5pt, fill=greencustom!70, opacity=0.6] coordinates {
  ($B_1$, 1.00) ($B_2$, 1.00) ($B_3$, 1.00)
};
\addplot+[bar shift=2.5pt, bar width=5pt, fill=blue!40, opacity=0.6] coordinates {
  ($B_1$, 0.17) ($B_2$, 0.76) ($B_3$, 0.78)
};
\addplot+[bar shift=7.5pt, bar width=5pt, fill=orange!60, opacity=0.6] coordinates {
  ($B_1$, 1.17) ($B_2$, 0.94) ($B_3$, 1.78)
};

\nextgroupplot[
  ylabel={Push Count Ratio},
  ymin=0, ymax=1.2,
]

\addplot+[bar shift=-2.5pt, bar width=5pt, fill=greencustom!70, opacity=0.6] coordinates {
  ($B_1$, 1.00) ($B_2$, 1.00) ($B_3$, 1.00)
};
\addplot+[bar shift=2.5pt, bar width=5pt, fill=blue!40, opacity=0.6] coordinates {
  ($B_1$, 0.14) ($B_2$, 0.72) ($B_3$, 0.51)
};
\addplot+[bar shift=7.5pt, bar width=5pt, fill=orange!60, opacity=0.6] coordinates {
  ($B_1$, 0.77) ($B_2$, 0.70) ($B_3$, 0.96)
};

\legend{Greedy-LM (baseline), H-LM, DS-LM}

\end{groupplot}
\end{tikzpicture}
\vspace{-2mm}
\caption{Relative performance of H-LM and DS-LM routing strategies with respect to the Greedy-LM baseline across benchmarks. Left: runtime; right: push count.}
\label{fig:bar_benchmark_dual}
\end{figure}
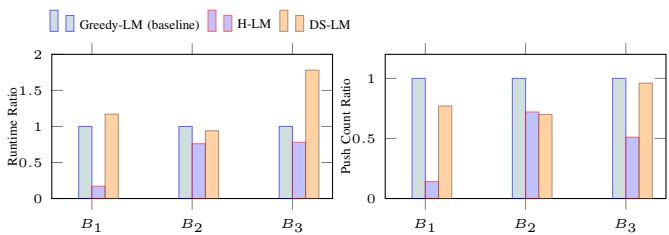

\begin{table}[ht]
\footnotesize
\renewcommand{\arraystretch}{0.9}
\setlength{\tabcolsep}{2pt}
\centering
\caption{Comparison between H-LM and DS-LM on Length-Matching $1 \times 6$ Multicasting Net.}
\label{tab:ds_lm}
\begin{tabular}{lccc ccc}
\hline
\multirow{2}{*}{\textbf{Benchmark}} & \multicolumn{3}{c}{\textbf{H-LM}} & \multicolumn{3}{c}{\textbf{DS-LM}} \\
\cmidrule(lr){2-4} \cmidrule(lr){5-7}
& Runtime (s) & Push & TWL & Runtime (s) & Push & TWL \\
\hline
$B_4 (1\times6\text{ }net)$ & 0.128 & 14229 & 529 & 0.123 & 12978 & 479 \\
\hline
\end{tabular}
\vspace{-1.0ex} 
\begin{flushleft}
\tiny \textsuperscript{*} Runtime measured on AMD Ryzen 7 7800X 12-core processor with 128~GB RAM.
\end{flushleft}
\end{table}

These findings indicate that while heuristic-based search strategies are generally more efficient, their effectiveness depends on the circuit topology. DS-LM introduces additional computational overhead but remains a valuable trade-off strategy for length-matching problems that require detours.
\subsection{Multi-pin Benchmarks}
As shown in Table~\ref{tab:ds_lm}, we compare the performance of DS-LM and H-LM in a specific application case ($B_4$): ten $1\times6$ OTTD-based multicasting beamforming nets with different combinations of target lengths as shown in Fig.~\ref{fig:benchmark}(f). Both routing approaches treat the network as a set of connection pairs. However, H-LM tends to generate excessive detours near the source, which subsequently blocks later routing targets. In contrast, DS-LM postpones detours until closer to the targets, resulting in a more efficient routing process with fewer nodes searched. However, the switch-function step of DS-LM still incurs slightly longer runtimes. Nonetheless, DS-LM achieves smaller total wire lengths (TWL), as nets tend to share paths near the sources and only diverge with distinct detours closer to the sinks. This demonstrates the advantages of DS-LM in applications where source congestion is a critical factor.



\section{Conclusion and Future Work}

Most photonic length-matching requirements arise within localized regions of a circuit, typically confined to the \textbf{`arms'} of photonic elements. Thus, computational complexity and runtime overhead remain manageable. In this work, we formally proved the completeness and heuristic admissibility of our proposed H-LM and DS-LM algorithms for two-pin length-matching routing problems under admissible heuristics. H-LM efficiently identifies feasible solutions by prioritizing early detours, whereas DS-LM delays detour routing near the destination, effectively addressing scenarios with source-region congestion and offering increased flexibility for programmable photonics. For multi-pin length-matching challenges, we further proposed a pin-ordering mechanism based on detour margins, reducing the likelihood of prematurely blocking feasible routes.

Current work is limited to solving the separate point-to-point length matching problem. In large-scale complex circuits on programmable PICs, congestion and thermal effects will happen, so combining this LM search with the adaptive congestion negotiation mechanism will be needed in future work. Due to the inherent symmetry and constraints of the hexagonal architecture, detours in length-matching routing are strictly quantized into increments of $+2$, $+4$, $+6$... following even integers. Prior work has analyzed the realizable path lengths in different architectures \cite{Gao204Provable}. Leveraging such architectural insights could enable the development of more informed heuristics and efficient pruning strategies, potentially reducing search complexity while preserving completeness. Exploring these architecture-specific optimizations remains a promising direction for future work.


\section*{Acknowledgment}
This work was made possible by collaborations with Ferre Vanden Kerchove and Prof. Mario Pickavet from IDLab, Ghent University, and Prof. Wim Bogaerts from the Photonics Research Group, Ghent University.
\bibliographystyle{IEEEtran}

\end{document}